\documentclass[conference]{IEEEtran}
\ifCLASSINFOpdf

\else

\fi
\usepackage{mymath}
\usepackage{amsmath}
\usepackage{balance}
\usepackage{graphicx}

\long\def\comment#1{}

\newfont{\bbb}{msbm10 scaled 700}

\newfont{\bb}{msbm10 scaled 1100}

\newcommand{\vect}[1]{\lowercase{\mathbf{#1}}}

\newcommand{\random}[1]{\uppercase{#1}}

\newcommand{\randomvect}[1]{\uppercase{#1}}

\newcommand{\xv}{\vect{x}}
\newcommand{\yv}{\vect{y}}

\newcommand{\Oc}{\alphab{O}} 

\newcommand{\Sc}{\alphab{S}} 

\newcommand{\Xc}{\alphab{X}} 
\newcommand{\Yc}{\alphab{Y}} 

\newcommand{\sr}{\random{s}}

\newcommand{\xr}{\random{x}}
\newcommand{\yr}{\random{y}}

\newcommand{\xrv}{\randomvect{x}}
\newcommand{\yrv}{\randomvect{y}}

\newcommand\ilr{\ensuremath{L_\beta}}
\newcommand\filr{\ilr\group{\xrv^n,P_\beta}}
\newcommand\filrop{\ilr\group{\xrv^n,P^{*}_\beta}}
\newcommand\filrox{\ilr\group{\randomvect{\hat{X}}^n,P_\beta}}

\newcommand\feasibleset{\ensuremath{\alphab{y}^n_\beta\group{s_0,\xv}}}

\newcommand\wlen{{(\beta+1)}/{\alpha}}

\newtheorem{theorem}{Theorem}
\newtheorem{definition}{Definition}

\begin{document}
%
\title{Smart Meter Privacy via the Trapdoor Channel}


\author{\IEEEauthorblockN{Miguel Arrieta$^*$ and I\~naki Esnaola$^{*\S}$}
\IEEEauthorblockA{$^*$Dept. of Automatic  Control and Systems Engineering, University of Sheffield, Sheffield S1 3JD, UK\\
 $^\S$Dept. of Electrical Engineering, Princeton University, Princeton, NJ 08544, USA\\
 \{marrieta2, esnaola\}@sheffield.ac.uk
}}


%


\maketitle

\begin{abstract}
A battery charging policy that provides privacy guarantees for smart meter systems with finite capacity battery is proposed. For this policy an upper bound on the information leakage rate is provided. The upper bound applies for  general random processes modelling the energy consumption of the user.
It is shown that the average energy consumption of the user determines the information leakage rate to the utility provider. The upper bound is shown to be tight by deriving the probability law of a random process achieving the bound.

\end{abstract}


%
\IEEEpeerreviewmaketitle

\section{Introduction}

The increasing appeal for  economical and environmentally-friendly energy calls for more efficient energy generation, distribution, and consumption \cite{ipakchi2009grid}. The introduction of a digital infrastructure  into the traditional power system takes steps towards this vision by providing a cyber layer that elevates the existing power system to a cyberphysical system. This advanced sensing and communication infrastructure envisioned by the smart grid enables high resolution and real time management of the processes within the grid. This application layer  also enables dynamic energy pricing, shifting user demand to match energy generation \cite{Hartway1999895}. Moreover, the introduction of energy consumption indicators through Smart Meters (SMs) are reported to reduce the energy consumption of the user by up to 15\% by raising awareness of the energy cost \cite{Wood2003821}.

While the high-resolution information provided by the smart grid brings clear advantages it also raises privacy and security concerns \cite{quinn2009privacy}, \cite{mcdaniel2009security,rouf2012neighborhood}. By analysing the consumption profile of a user, techniques such as non-intrusive load monitoring (NILM) \cite{hart1992nonintrusive} track and recognise appliance usage patterns \cite{McKenna2012807,Molina-Markham:2010:PMS:1878431.1878446}. Human presence, usage of individual appliances \cite{Enev:2011:TVP:2046707.2046770,985144}, and even tuned TV channel \cite{greveler2012multimedia} are among the list of recognizable elements \cite{4266955}. This privacy breach hinders the implementation of some of the essential components of the smart grid \cite{quinn2009privacy}, \cite{mcdaniel2009security}. Within this paradigm, SMs are central components to the dilemma posed by the need for accurate monitoring while providing privacy. In 2009 two bills law aimed to enforce the usage of SMs were blocked by the Senate of the Netherlands motivated by the privacy concerns that emerge as a result of the increased penetration of SMs \cite{Cuijpers2013}.

There is a growing body of literature addressing the conflict between efficient energy monitoring and privacy brought forth by the introduction of SMs. In \cite{6007070,6102315} obfuscation of the knowledge that the utility provider (UP) has about the energy consumption of the user is studied. Indeed, in the case in which the SM readings are the only source of information available to the UP, obfuscation yields some degree of privacy. Obfuscation is achieved by several  mechanisms, such as aggregating the consumption of multiple users \cite{6007070}, compression of the energy consumption sequences \cite{6102315} or homomorphic encryption \cite{Li2010} among others. A different approach to the problem arises in the setting in which users have access to alternative energy sources or energy storage devices \cite{zhu2017privacy}. In this case, the UP has perfect knowledge of the energy provided to the user, but the user employs the energy storing capability of the system to dissociate the energy consumed by the appliances from the energy provided by the UP. In \cite{gomez2013privacy, tan2012smart, li2015privacy} the case in which the user is assumed to have an alternative energy source with instantaneous power constraints is studied. In \cite{tan2012smart,varodayan2011smart, kalogridis2010privacy, 6003811} the user is assumed to have access to a rechargeable battery and the energy consumption is modelled as an independent and identically distributed (i.i.d.)  random process.

In a practical setting, the energy consumption profiles of users exhibit non-stationary statistical structures and are not well described by memoryless random processes \cite{kalogridis2010privacy}.  Moreover, information-theoretic privacy measures for random processes that are not i.i.d. are still not well understood \cite{6620443}. The privacy utility tradeoff is characterized for stationary Gaussian energy consumption models in \cite{65eac443f7d6420a9bb100e3a77b70a6}. A first-order time-homogeneous Markov chain is considered in \cite{7536745}. In  this work we adopt a non-probabilistic framework by modelling the energy management system (EMS) with a finite capacity battery as a finite state channel without probabilistic structure. Inspired by the code construction in \cite{ahlswede1987optimal} we propose an energy charging policy and we characterize the privacy guarantees of the strategy for general random processes. We also particularize the analysis to the case in which the average energy consumption of the user is known. For this case we provide an upper bound to the amount of information that the user leaks to the UP and show that the average energy consumption governs the privacy that is achievable by the user. 
In this paper vectors are denoted by bold font, e.g. $\xv$, random variables are denoted by upper-case, e.g.  $\xr$, and vectors of $n$ random variables are denoted by super-indexing the size, e.g. $X^n$.

\section{Battery System Model}
\label{sec:sys_model}

We consider the energy management system (EMS) operating in discrete-time illustrated in Fig. 1.
The energy consumption of the user is modelled as a discrete-time  random process $\xr_1,\xr_2, \ldots$ where $\xr_i$ is a random variable taking values in $\Xc = \{0,1, \ldots,\alpha\}$. In this setting $\xr_i$ describes the energy consumed at time instant $i\in\mathbb{N}$.
At each time instant, the energy consumption of the user is satisfied by requesting energy from the UP or by discharging the battery. This decision is taken by the energy management unit (EMU) at each time instant based on some power consumption performance and privacy criteria. 
The energy requested from the UP is also a discrete-time  random process $\yr_1,\yr_2,...$ where $\yr_i$ is a random variable taking values in $\Yc = \{0,1,...,\gamma\}$ and describing the energy requested from the UP at time instant $i\in\mathbb{N}$. When $\yr_i > \xr_i$, the excess energy is stored in the battery. Alternatively, when $\yr_i \leq \xr_i$, the energy deficit is obtained from the battery. %
We assume the UP is able to satisfy the energy consumption of the user even in the case when there is no battery, \ie\ $\gamma \geq \alpha$.

\begin{figure}[t]
  \centering
  \includegraphics[width=\columnwidth]{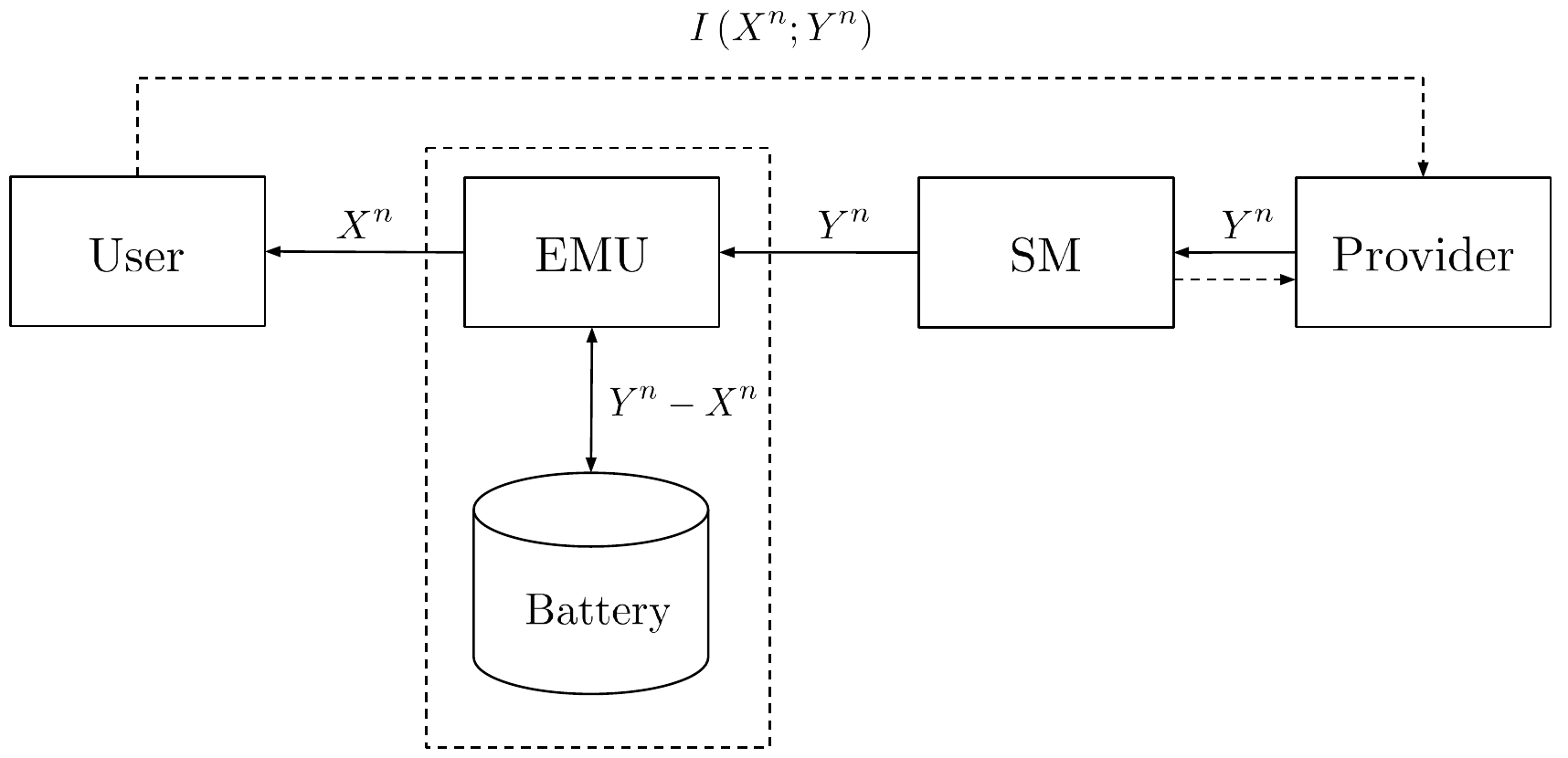}
  \label{fig:battery_system}
  \caption{Diagram of an Energy Management System with finite battery.}
\end{figure}

In the above model, the state $\sr_i$ taking values in $\alphab{s} = \{0,1,...,\beta\}$ describes the energy stored in the battery at time $i$. The energy stored in the battery $\sr_i$ is  a function of the previous energy consumption $\xrv^i$, the energy request $\yrv^i$, and the initial state $s_0 \in \alphab{S}$ given by
\begin{equation} \label{eq:battery_filing}
  \sr_i = s_0+ \sum_{k=0}^{i-1} \group{ \yr_k - \xr_k }.
\end{equation}
Within this setting, a power outage occurs  when $\sr_i + { \yr_i - \xr_i } < 0$, and energy is wasted when $\sr_i + { \yr_i - \xr_i } > \beta$. In the following we focus on EMUs that do not allow power outages nor energy wasting but provide a quantifiable privacy guarantee for the user. Given a particular realization $\xv\in\Xc^n$ of the random process $\xr_1, \xr_2, \ldots, \xr_n$ modelling the energy consumption of the user up to time $n$, the set of energy requests that the EMU can implement is limited by the power outage and the energy waste constraints. The following definition describes the set of energy requests that the EMU can implement.

\begin{definition}
Given an energy consumption sequence $\xv \in \Xc^n$ as the input of an EMU with a battery of capacity $\beta$, the {\it set of stable energy request sequences} that avoid power outages and energy waste is given by
\begin{equation}
\label{eq:stable_battery_outputs}
  \feasibleset \defeq \{ \yv\in\Yc^n : s_i+y_i - x_i \in \alphab{s} \textrm{ for all }i\},
\end{equation}
where $s_i \in \Sc_i$ is the state of the battery at time $i$, determined by $\xv$ and $\yv$ according to (\ref{eq:battery_filing}).
\end{definition}
The task of the EMU is therefore to choose a particular sequence in the $\feasibleset$ for a given power consumption realization $\xv$. The structure of the particular choice determines the policy implemented by the EMU and is captured by the following definition.
\begin{definition}
Given an EMU with a battery of capacity $\beta$ the {\it set of stable battery policies} is the set of mappings between the energy consumption sequences and the set of stable energy request sequences given by
\begin{equation} \label{eq:stable_battery_policies}
  \alphab{P}_\beta \defeq \{ P_{\beta} : \alphab{s}\cartesianProd \Xc^n \to \feasibleset\}.
\end{equation}
\end{definition}
Since $\yrv^n$ is known by the UP, the information about the energy consumption of the user that the UP acquires via the energy request is given by the mutual information $I(\xrv^n; \yrv^n)$. As in \cite{varodayan2011smart} this motivates the following definition of privacy.

\begin{definition}
Given an EMU operating with the stable battery policy $P_\beta$, the information about the consumption of the user, $\xrv^n$, that is leaked to the UP is the \emph{information leakage rate} given by
\begin{equation}
\label{eq:ILR}
    \filr \defeq \frac{1}{n} I(\xrv^n; \yrv^n).
\end{equation}
\end{definition}

In a SM privacy context, the aim of the EMU is to choose a stable battery policy $P_\beta$ that minimizes the information leakage rate, i.e. maximizes the privacy of the user. Note that the information leakage depends on the joint probability distribution of $\xrv^n$ and $\yrv^n$. In general, the evaluation of (\ref{eq:ILR}) yields involved expressions that are difficult to evaluate \cite{6620443}. For that reason previous results \cite{gomez2013privacy,tan2012smart,varodayan2011smart} tend to consider simple probabilistic models, e.g. memoryless processes, to evaluate the mutual information. In the remaining of the paper we analyze the privacy guarantees for general discrete-time random processes modelling the user consumption. To that end, we model the EMS with a battery of capacity $\beta$ as a non-probabilistic finite-state channel \cite{ahlswede1987optimal}. The rationale for this approach and the equivalence between the EMS and a non-probabilistic channel are discussed in the following section.


\vspace{-2mm}

\subsection{Equivalence with the trapdoor channel}

\begin{figure}[t]
  \centering
  \includegraphics[width=\columnwidth]{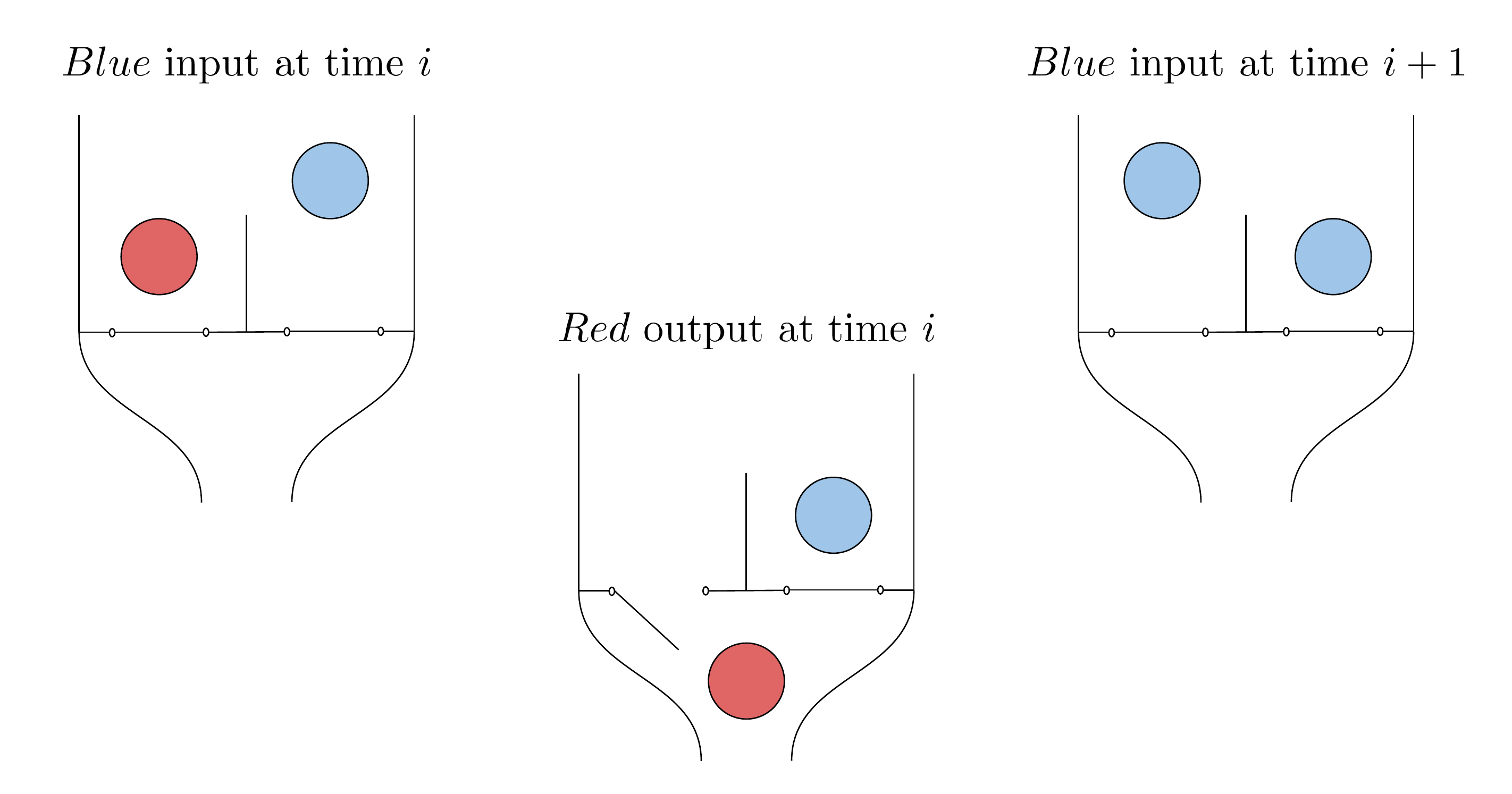}
  \caption{\!\!\!\cite{robert1990ash} Diagram depicting the operation of a trapdoor channel with $\beta = 1$.}
\end{figure}

The \emph{trapdoor channel} \cite{ahlswede1987optimal,robert1990ash} is defined as a box containing $b_0$ \emph{blue} balls and $\beta-b_0$ \emph{red} balls. The operation of the channel is depicted in Fig. 2.  At time $i$ a new ball $\xr_i$ coloured \emph{blue} or \emph{red} is thrown into the box. Immediately after, one of the $\beta+1$ balls inside the box is selected and taken out of the box. Let $\yr_i$ denote the ball extracted at time $i$. Following this model, the number of \emph{blue} balls inside the box at time $i$ is given by
\begin{equation} \label{eq:trapdoor_filing}
  b_i = b_0 + \sum_{k=0}^{i-1}\typedcard{\xr_k}{bl} - \sumrange{k}{0}{i-1} \typedcard{ \yr_k }{bl},
\end{equation}
where the indicator function $bl(\cdot)$ equals $1$ when its argument is coloured \emph{blue} and $0$ otherwise. Similarly, the number of \emph{red} balls inside the box at time $i$ is given by $r_i = \beta - b_i$. Replacing $b_i = \beta - r_i$ and $b_0 = \beta - r_0$ into (\ref{eq:trapdoor_filing}) yields:
\begin{equation} \label{eq:trapdoor_filing2}
  r_i = r_0 + \sum_{k=0}^{i-1} \big( \typedcard{\yr_k}{bl} - \typedcard{\xr_k }{bl} \big).
\end{equation}
The number of \emph{red} balls inside the box is bounded between $0$ and $\beta$. For a box of capacity $\beta$ the set of \emph{stable output balls} $\alphab{y}^n_\beta\group{r_0,\xv}$ is defined as the set of outputs $\yv^n \in \Yc^n$ than can be pulled out of the box given an initial state $r_0$ and an input sequence $\xv\in\Xc^n$, \ie\
\begin{equation} \label{eq:stable_trapdoor_outputs}
  \alphab{y}^n_\beta\group{r_0,\xv} = \{ \yv : r_i+\typedcard{y_i}{bl}-\typedcard{x_i}{bl} \in \alphab{R} \textrm{ for all }i\},
\end{equation}
where $\alphab{R}=\{0,...,\beta\}$. 

It is easy to see that for the case in which $\gamma = \alpha = 1$ the EMS with a battery of capacity $\beta$ described in Section \ref{sec:sys_model} is equivalent to the \textit{trapdoor channel} of capacity $\beta$. The set of balls inside the box determines the state of the trapdoor channel, and similarly, the amount of energy stored in the battery determines the state of the EMS channel. For the case in which $\alpha=\gamma=1$ both systems are equivalent since requesting energy from the grid corresponds to extracting a ball from the trapdoor channel. Similarly, replacing a ball from the trapdoor channel in (\ref{eq:trapdoor_filing2}) corresponds to charging the battery of the EMS in (\ref{eq:battery_filing}).

\section{Privacy With An Arbitrary Energy Consumption}
\label{sec:pr_arbit}

In this section, we provide bounds on the information leakage rate when no restrictions are imposed on the probability law of $\xrv^n$. We first propose the construction of a \emph{stable battery policy} $P^*_\beta$ and  characterize an upper bound on the information leakage rate \ilr\ induced by $P^*_\beta$ and any arbitrary random process $\xrv^n$. Furthermore, we show the tightness of the upper bound by constructing a random process $\randomvect{\hat{x}}$ whose leakage is tight with respect to the upper bound. Moreover, the leakage rate induced by the random process $\randomvect{\hat{x}}$ is shown to be independent of the employed battery policy $P_\beta \in \alphab{P}_\beta$. This shows that the upper bound is tight with respect to the minimum \emph{information leakage rate} a \emph{stable battery policy} $P_\beta$ can guarantee for general random processes $\xrv^n$.

The approach to policy construction in this section is similar to the code construction in \cite{ahlswede1987optimal} where a trapdoor channel with a box of size $\hat{\beta}$ is considered. Therein, at every time instant a ball numbered $1,2,...,\hat{\alpha}$ is introduced into the box and one of the $\hat{\beta}+1$ inside the box is extracted. In \cite[Section II]{ahlswede1987optimal} the case in which the box acts as a jammer trying to obstruct the communication process between a sender inserting the balls into the box and a receiver drawing the output balls is studied. Therein, the ball extracted from the box is selected in order to minimize the mutual information between the input and the output. Note that the extraction criteria is not probabilistic  and is instead analyzed using combinatorial tools. Remarkably, in \cite[Proposition 1]{ahlswede1987optimal} it is shown that the \emph{Shannon capacity} $C_{\hat{\beta}}$ of such channels is lower bounded by
\begin{equation}
  C_\beta \geq \frac{\log \hat{\alpha}}{\hat{\beta}+1}.
\end{equation}
Moreover, when $\hat{\alpha}=2$ the capacity is upper bounded by
\begin{equation}
  C_\beta \leq \frac{1}{\hat{\beta}+1}.
\end{equation}

In \cite{ahlswede1987optimal} the output is constrained to permutations of the input sequence. In our setting, the sum of the output sequence is bounded by the sum of the input sequence and the size of the battery. However, the output is not required to contain the same symbols as the input . For that reason, the approach in \cite{ahlswede1987optimal} requires some modification but the main idea remains. The derivation is presented in the next section.


\subsection{Upper bound on the information leakage rate}
We propose a battery policy based on the code construction in \cite{ahlswede1987optimal}. The codebook proposed in the trapdoor channel context is the counterpart of the battery policy in the smart meter case. The proposed policy structures the energy request sequences according to the output alphabet defined below.

\begin{definition}
Consider the set of codewords of length $l$ constructed by repetitions of $0$ or $\alpha$ symbols, \ie\ $\alphab{o}_l = \{(0,0,\cdots,0),(\alpha,\alpha,\cdots,\alpha)\}$. For $n=lm$, we define the {\it block repetition alphabet} as the set $\alphab{O}^m_l$ of sequences obtained by the $m$-fold concatenation of codewords of length $l$. Specifically
\begin{equation}
    \alphab{o}^m_l = \alphab{o}_l\cartesianProd\alphab{o}_l\cartesianProd...\cartesianProd\alphab{o}_l.
\end{equation}
\end{definition}
We now define a stable policy that maps the energy consumption of the user to the output sequences constructed with the {\it block repetition alphabet} $ \alphab{o}^m_l $.
 
\begin{definition}
\label{def:BBP}
A \emph{block battery policy} $P_\beta^*$ is a mapping of the form
\begin{equation}
    P^*_\beta : \alphab{S}\cartesianProd\alphab{x}^n \to \alphab{o}^m_l \cap \alphab{y}^l_\beta\group{s_0,\xv}.
\end{equation}
\end{definition}
Note that a block battery policy is nothing more than a strategy to assign to each input sequence a stable energy request sequence constructed with a block repetition alphabet. With these definitions at hand we now provide the following privacy guarantee. 
\begin{theorem}
\label{th:UB}
    Consider an EMS with a battery of capacity $\beta$ and initial state $s_0 \in \Sc$. Let $\xrv^n$ be a random process with $\xr_i$ taking values in $\Xc=\{0,1,...,\alpha\}$ for $i=1,2, \ldots n$ and $P^*_\beta$ a \emph{block battery policy} as described in Definition \ref{def:BBP}. Then for $l \leq \floor{\wlen}$ at least one policy $P^*_\beta$  exists such that 
    \begin{equation}
        \filrop \leq \oo{\floor{\wlen}}.
    \end{equation}
\end{theorem}

\begin{proof} 
Notice that the information leakage rate is upper bounded by
\begin{equation}
    \filrop = \frac{1}{n} \muinf{\xrv^n}{\yrv^n} \leq \frac{1}{n} \ent{\yrv^n}.
\end{equation}
Since $\yrv^n$ takes values in $\alphab{O}^m_l$ and $|\alphab{O}^m_l|=2^m$ the following  holds:
\begin{equation}
    \frac{1}{n} \ent{\yrv^n} \leq \frac{1}{n} \log|\alphab{O}^m_l| =  \frac{1}{n} \log\group{2^m} = \frac{m}{n} = \frac{1}{l}.
\end{equation}

We now show that when $l \leq \wlen$ there exists at least one \emph{block battery policy} $P^*_\beta$ for every initial state $s_0$ and consumption $\xv$. To prove this we establish that for every realization $\xv$ and initial state $s_0$ there exist an energy request sequence determined by $\yv \in \alphab{O}^m_l$ such that $\yv$ belongs to the set of stable energy requests \feasibleset. The strategy is to notice that $\alphab{O}^m_l \cap \feasibleset \not= \emptyset$ for $m=1$ and to then prove by induction that the non-emptiness of the intersection holds for $m\geq1$.

The intersection $\{(0,0,\cdots,0),(\alpha,\alpha,\cdots,\alpha)\} \cap \alphab{y}^l_\beta\group{s_0,\xrv^l}$ is non-empty if and only if either the sequence $(0,0,\cdots,0)$ or $(\alpha,\alpha,\cdots,\alpha)$ belong to \feasibleset. Jointly with (\ref{eq:stable_battery_outputs}) this implies that either
\begin{equation} \label{eq:y_is_0}
       s_i+0 - x_i \in \alphab{s}
\end{equation}
or
\begin{equation} \label{eq:y_is_lambda}
    s_i+\alpha - x_i \in \alphab{s}
\end{equation}
holds for $i \leq l$. In the first case, described in (\ref{eq:y_is_0}), we have that $0 - x_i \leq 0$ for $i=0,\cdots,l-1$.  Hence, the energy stored in the battery, $s_i$, decreases monotonically. Therefore, all $s_i$ belong to $\alphab{s}$ when $s_i \geq 0$ on the last time step, i.e.,
\begin{equation} \label{eq:y2_is_0}
    0 \leq s_0 - \sum_{i=0}^{l-1}x_i.
\end{equation}
Similarly, in the case described by (\ref{eq:y_is_lambda}), we have that $\alpha-x_i \geq 0$ and the  energy stored increases monotonically. It is then sufficient to show that
\begin{equation} \label{eq:y2_is_lambda}
    s_0 - \sum_{i=0}^{l-1} x_i \leq \beta - \alpha l.
\end{equation}
When $\beta - \alpha l \geq -1$ every integer $s_i$ satisfies at least one of the inequalities given by (\ref{eq:y2_is_0}) and (\ref{eq:y2_is_lambda}). This ensures that either (\ref{eq:y2_is_0}) or (\ref{eq:y2_is_lambda}) hold for every $s_0 \in \alphab{s}$ and $\xv \in \alphab{x}^l$, and therefore, the intersection $\alphab{o}^m_l \cap \feasibleset$ is non-empty. This completes the proof for $m=1$. The induction for $m\geq1$ is straightforward as the proof for $m=1$ holds for every initial state $s_0$. \end{proof}

The upper bound derived in Theorem \ref{th:UB} is depicted for different battery sizes in Fig. 3. It is interesting to note that the privacy guarantees increase significantly for small values of $\beta/\alpha$ but the benefit vanishes as the size of the battery increases. 

\begin{figure}[t]
\label{fig:th1}
  \centering
  \includegraphics[width=\columnwidth]{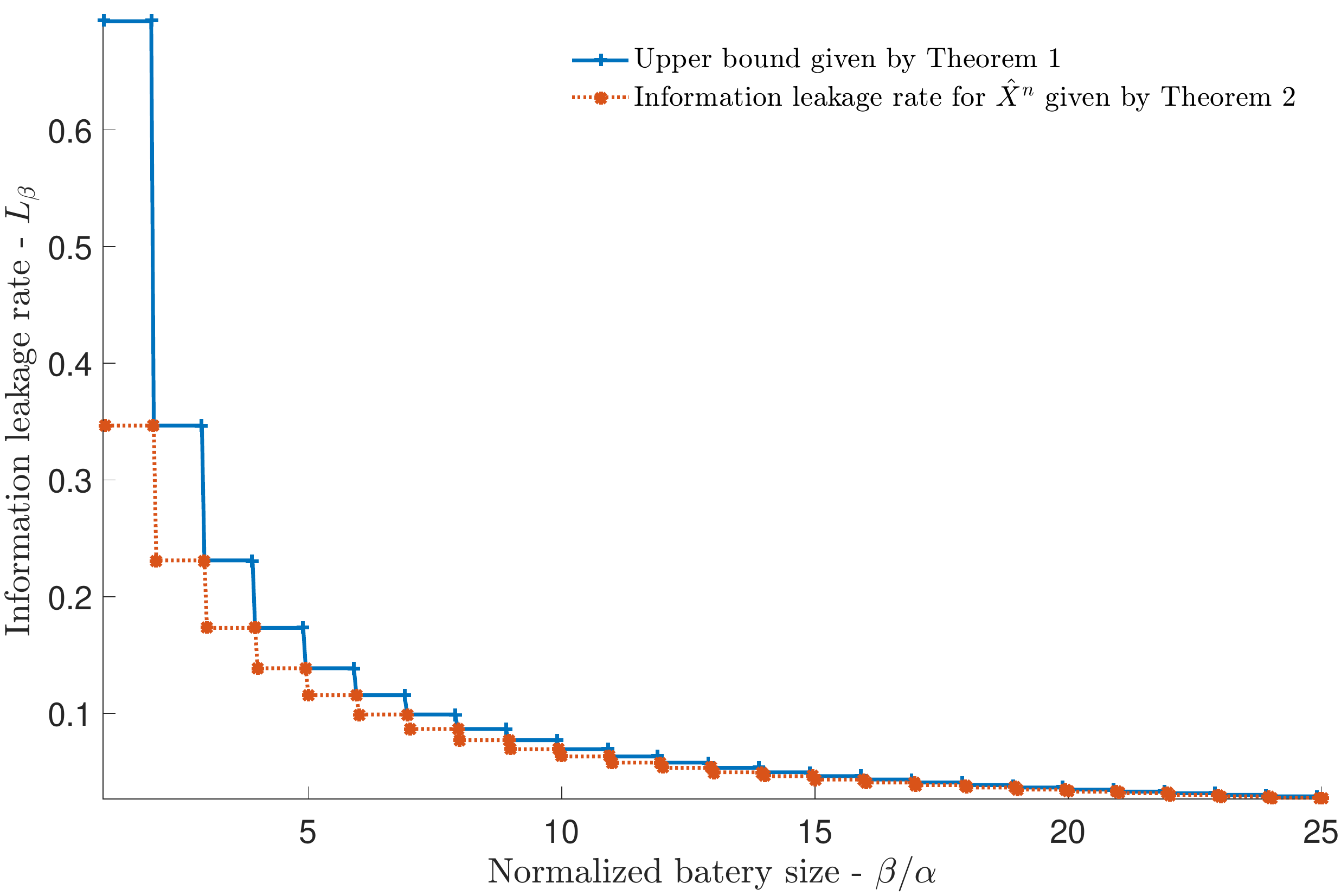}
  \caption{Upper bound on the information leakage rate of an EMS as a function of the ratio between the battery size and the peak power consumption. }
\end{figure}

\subsection{Tightness of the upper bound}

We now study the tightness of the upper bound presented in Theorem \ref{th:UB}. To this end, we construct a random process modelling the energy consumption of the user that is tight with respect to the result in Theorem \ref{th:UB} for every battery policy $P_\beta \in \alphab{P}_\beta$.

\begin{theorem}
\label{th:worst_case}
Consider an EMS with a battery of capacity $\beta$ and initial state $s_0$. Let $\randomvect{\hat{x}}^n$ be a random process taking uniformly distributed values in  $\alphab{O}^m_l$ with $l=\ceil{\wlen}$. Let $P_\beta$ be a \emph{stable battery policy}. Then
    \begin{equation}
        \filrox = \frac{1}{\ceil{\wlen}}.
    \end{equation}
\end{theorem}

\begin{proof}
    We expand $L_\beta$ as
    \begin{equation}
      \frac{1}{n} \muinf{\xrv^n}{\yrv^n} = \frac{1}{n}\ent{\xrv^n} - \frac{1}{n}\ent{\xrv^n|\yrv^n}.
    \end{equation}
   When $\xrv^n$ is uniformly distributed over the alphabet $\alphab{O}^m_l$ it yields
    \begin{equation}
      \frac{1}{n}\ent{\xrv^n} = \frac{1}{n}m = \frac{1}{l}.
    \end{equation}
   We now show that the equivocation rate $\frac{1}{n}\ent{\xrv^n|\yrv^n}$ is $0$ when $\xrv^n$ takes values in $\alphab{o}^m_l$ with $l>{\beta}/{\alpha}$. We prove by induction that when the input realization $\xv$ belongs to $\alphab{O}^m_l$ with $l>{\beta}/{\alpha}$, the sets \feasibleset\ of stable output words generated by different consumption sequences are disjoint, \ie
    \begin{equation}
    	\alphab{y}^n_\beta\group{s_0,\hat{\xv}} \cap \alphab{y}^n_\beta\group{s_0,\xv} = \emptyset \textrm{ for } \hat{\xv} \not= \xv.
    \end{equation}
As a result, any request sequence $\yv \in \alphab{y}^n_\beta\group{s_0,\xv}$ unequivocally determines the generating input $\xv$. In other words, given an output sequence $\yv$ there is no uncertainty about the input $\xv$, and therefore, the equivocation rate $\frac{1}{n}\ent{\xrv^n|\yrv^n}$ is $0$.
    
    For $m=1$ there are two possible inputs $(0,0,\cdots,0)$ and $(\alpha,\alpha,\cdots,\alpha)$. When $\xv=(0,0,\cdots,0)\in\Oc^1_l$ the energy stored in the battery at time $l$ is given by
    \begin{equation} \label{eq:any_lw1}
      s_l = s_0 + \sum_{i=0}^{l-1}\group{y_i - 0 }.
    \end{equation}
    Similarly, when $\xv =(\alpha,\alpha,\cdots,\alpha)\in\Xc^l$ the energy stored in the battery at time $l$ is given by
    \begin{equation} \label{eq:any_lw2}
      z_l = s_0 + \sum_{i=0}^{l-1}\group{y_i - \alpha }.
    \end{equation}
    Taking the difference between (\ref{eq:any_lw1}) and (\ref{eq:any_lw2}) yields:
    \begin{equation}
      s_l - z_l = \sum_{i=0}^{l-1}\alpha = l\alpha.
    \end{equation}
    When $z_l \in \alphab{s}$ we have that $s_l = z_l + l\alpha \geq l\alpha$, showing that for $l\alpha > \beta$ the events $z_l\in\alphab{s}$ and $s_l\in \alphab{s}$ do not occur simultaneously. This implies that the set of output words belonging to $\alphab{y}^n_\beta\group{s_0, (0,0,\cdots,0)}$ and $\alphab{y}^n_\beta\group{s_0,(\alpha,\alpha,\cdots,\alpha)}$ is empty for every initial state $s_0$. Therefore the sets are disjoint and $\frac{1}{n}\ent{\xrv^n|\yrv^n} =0$. The proof for $m > 1$ follows by induction and noticing that the proof above is valid for every initial state $s_0$.
    %
\end{proof}

\section{Privacy With An Average Energy Constraint}
The information leakage rate bounds provided in Section \ref{sec:pr_arbit} do not impose any moment restriction on the random process modeling the energy consumption of the user. Indeed, they depend only on the range of the energy consumption and on the size of the battery. However, one of the most widely used energy consumption metrics is the average energy consumption over an arbitrary time interval. In fact, it is common for SMs to display this information to the user.
In the following, we particularize the results in Theorem \ref{th:UB} and Theorem \ref{th:worst_case} to the case in which the average energy consumption of the user is specified. Specifically, we analyze the impact of the average energy consumption on the privacy performance. We define the average energy consumption of the random process $\xrv^n$ as 
\begin{equation}
    \mu_n = \avg{ \frac{1}{n}\sum_{i=0}^{n-1}{\xr_i}}.
\end{equation}
Note that since we do not impose any stationarity condition on the random process $\xrv^n$, the average energy consumption is a function of the time index $n$. This agrees with the non-stationary nature observed in energy consumption profiles of users \cite{kalogridis2010privacy}.
\begin{figure}[t!]
  \centering
  \includegraphics[width=\columnwidth]{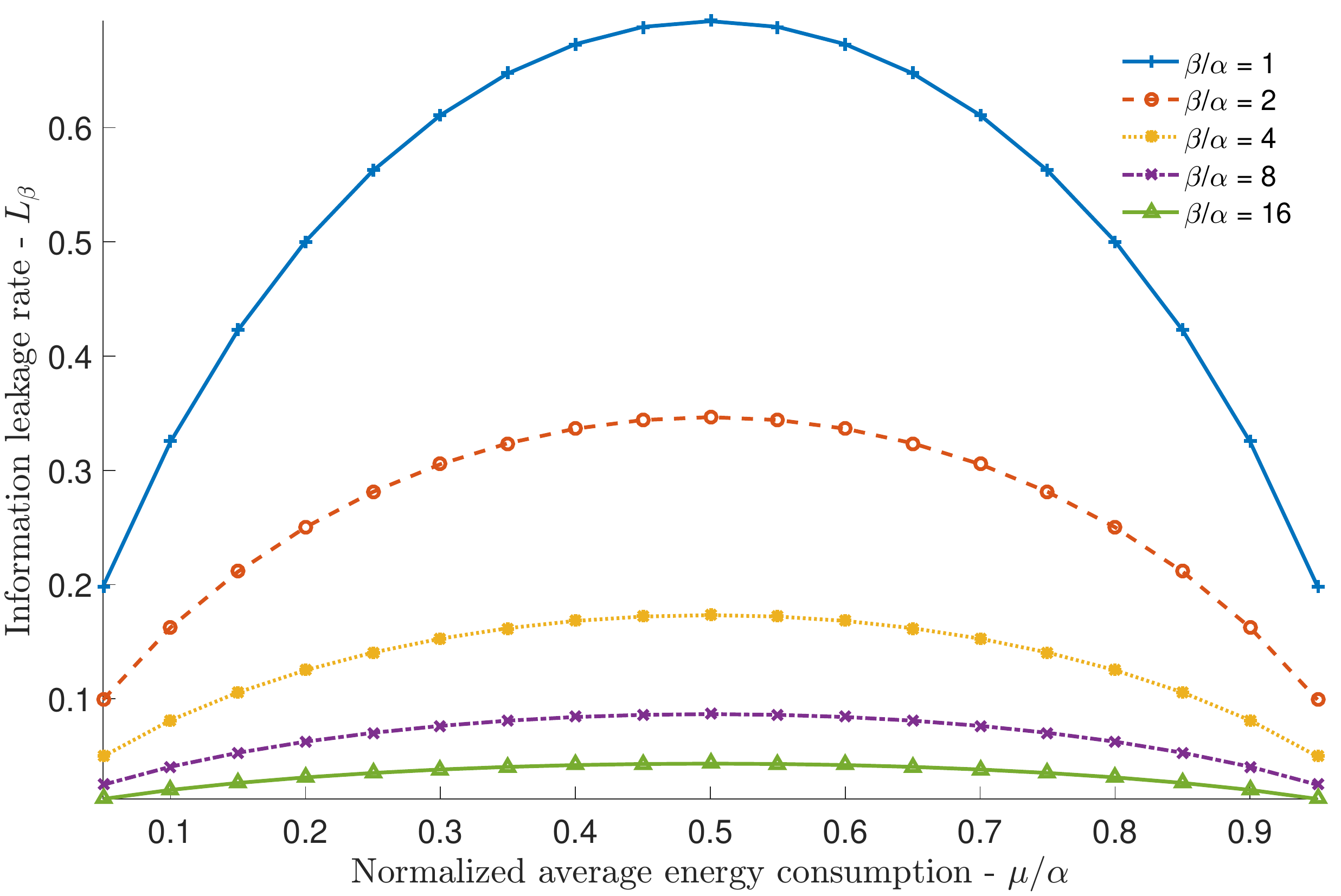}
  \caption{Upper bound on the information leakage rate of an EMS when $n\rightarrow\infty$ as a function of the average energy consumption of the user for different values of the ratio between the battery size and the peak power consumption.}
\end{figure}

\subsection{Upper bound on the information leakage rate}
The following result provides an upper bound on the  information leakage rate for random processes $\xrv^n$ with average energy consumption $\mu_n$.
\begin{theorem}
    \label{th:avg_UP}
    Consider a battery system with capacity $\beta$ and initial state $s_0$. Let $\xrv^n$ be a random process with \emph{average energy consumption} $\mu_n$. Let $P^*_\beta$ be a \emph{block battery policy}, then for $l \geq \floor{\wlen}$ at least one policy ${P}^*_\beta$ exists such that
    \begin{equation}
        \filrop \leq \frac{\max\left(\!\!H_2\!\group{\frac{\mu_n-\frac\beta n}{\alpha}},H_2\group{\frac{\mu_n+\frac\beta n}{\alpha}}\right)}{\floor{\wlen}},
    \end{equation}
 where $H_2(p)=-p\log_2p-(1-p)\log_2(1-p)$ denotes the binary entropy.   
\end{theorem}

\begin{figure}[t!]
  \centering
  \includegraphics[width=\columnwidth]{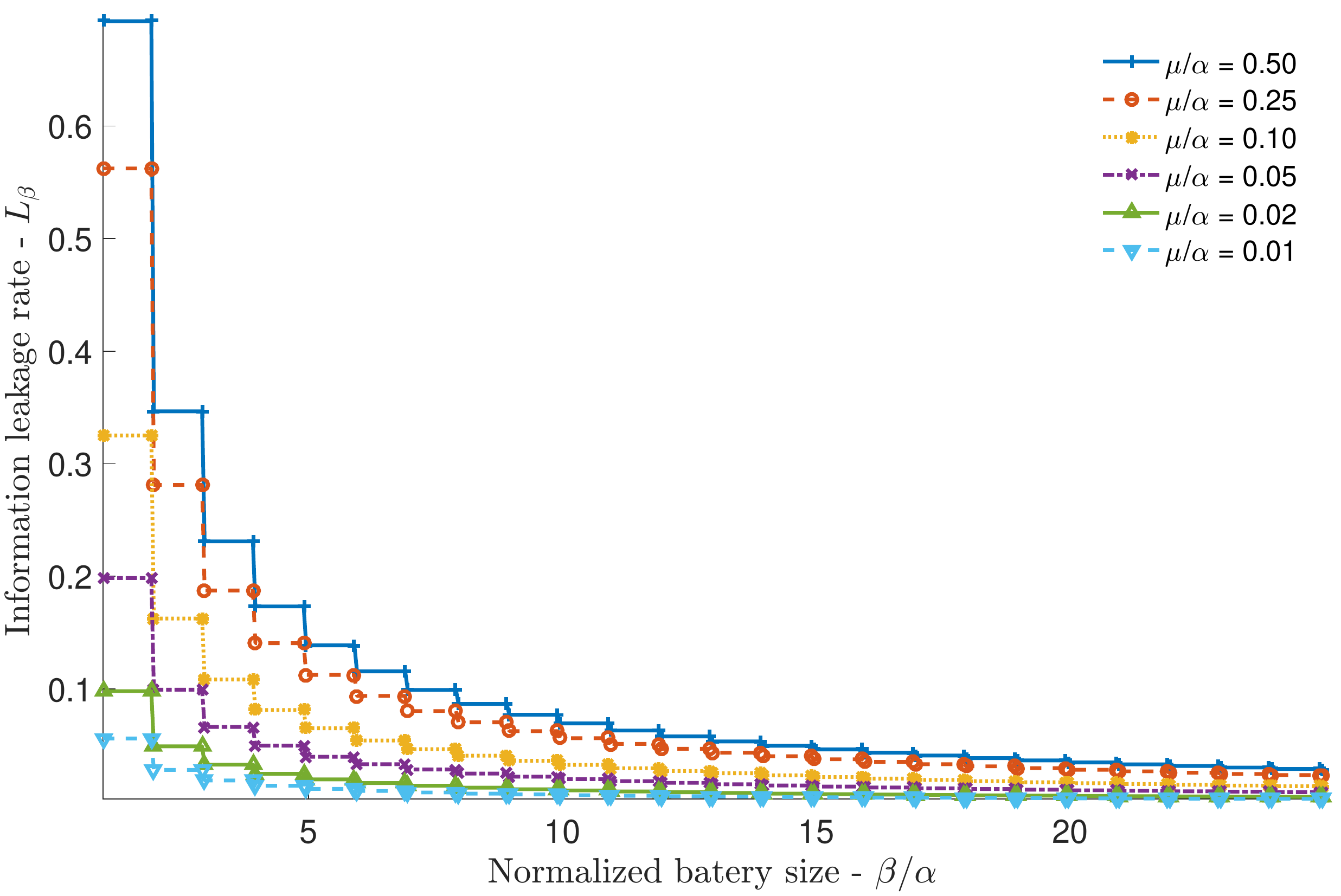}
  \caption{Upper bound on the information leakage rate of an EMS when $n\rightarrow\infty$ as a function of the ratio between the battery size and the peak power consumption for different values of the average energy consumption of the user}
\end{figure}

\begin{proof}
The entropy of a random process $\yrv^n$ taking values in $\alphab{O}^m_l$ is upper bounded by
\begin{IEEEeqnarray}{lCl}
	\label{eq:chain}
        \frac{1}{n}\!H\!\left(\yrv^n\right)\!&=&\!\frac{1}{n}\!\sum_{i=0}^{m-1}\!\!H\!\left( \yr_{il}, \ldots, \yr_{(i+1)l-1}|\yr_{0}, \ldots, \yr_{il-1}\right) \\
        \label{eq:cond}
         &\leq& \frac{1}{n}\!\! \sum_{i=0}^{m-1}\ent{\yr_{il}, \ldots, \yr_{(i+1)l-1}},
\end{IEEEeqnarray}
where (\ref{eq:chain}) follows by applying the chain rule and (\ref{eq:cond}) follows from the fact that conditioning reduces entropy. Notice that (\ref{eq:cond}) is the entropy of $m$ sequences $\yrv^l$ taking values in $\alphab{O}_l$, and therefore, the entropy of $\yrv^n$ is upper bounded by
\begin{equation}
\label{eq:ent_in_alph}
\frac{1}{n} \ent{\yrv^n}\leq \frac{1}{l}H_2\left(\frac{\avg{\frac{1}{n}\sum_{i=0}^{n-1}\yr_i}}{\alpha} \right),
\end{equation}
for the case in which each sequence $\yrv^l$ is independent and identically distributed, i.e. with distribution
\begin{equation}
\mathbb{P}\left[\yrv^l=(\alpha,\alpha, \ldots, \alpha)\right]=\frac{\avg{\frac{1}{n}\sum_{i=0}^{n-1}\yr_i}}{\alpha},
\end{equation}
and 
\begin{equation}
\mathbb{P}\left[\yrv^l=(0, 0, \ldots, 0)\right]=1-\mathbb{P}\left[\yrv^l=(\alpha,\alpha, \ldots, \alpha)\right].
\end{equation}
 We now bound the average energy requested from the grid as a function of the average energy consumption of the user and the battery size. Dividing (\ref{eq:battery_filing}) by $n$ and taking the expected value yields
    \begin{equation} \label{eq:avg_energy_constrain_transfer}
        \avg{\frac{1}{n}\sum_{i=0}^{n-1}\yr_i} = \avg{\frac{1}{n}\sum_{i=0}^{n-1}\xr_i} + \avg{\frac{\sr_n-s_0}{n}},
    \end{equation}
or equivalently
\begin{equation}
\label{eq:avg_bounds}
\mu_n - \frac{\beta}{n}\leq \mathbb{E}\left[\frac{1}{n}\sum_{i=0}^{n-1}\yr_i\right]\leq  \mu_n + \frac{\beta}{n}.
\end{equation}    
Notice now that for $l\leq\wlen$, and for every initial state $s_0 \in \alphab{s}$ and input realization $\xv \in \alphab{x}^n$ there exists a sequence $\yv \in \alphab{O}^m_l$ such that $\yv$ belongs to the set of \emph{stable energy requests} \feasibleset. This completes the proof.
\end{proof}

The upper bound on the information leakage rate when the average energy consumption of the user is known and $n\rightarrow\infty$ is illustrated in Fig. 4 and Fig. 5. As expected, the binary entropy term in Theorem 3 introduces concavity in the upper bound as shown in Fig. 4. Interestingly, Fig. 5 shows that the information leakage rate reduction as the size of the battery increases is less significant for extreme values of the average energy consumption.

\subsection{Tightness of the upper bound}

Proceeding in a similar fashion as in Section \ref{sec:pr_arbit} we now prove that the upper bound in Theorem \ref{th:avg_UP} is tight for a certain class of random processes modelling the energy consumption.

\begin{theorem}
  Consider a battery system with capacity $\beta$ and initial state $s_0$. Let $\randomvect{\hat{x}}^n$ be a random process with \emph{average energy consumption} $\mu_n$ and taking values in $\alphab{o}_l^m$ with $l = \ceil{\wlen}$. Let ${P}_\beta$ be a stable battery policy, then 
  \begin{equation}
      \filrox = \frac{1}{\ceil{\wlen}} H_2\group{\frac{\mu_n}{\alpha}}.
  \end{equation}
\end{theorem}

\begin{proof}
  Borrowing from (\ref{eq:ent_in_alph}) the entropy rate of the random process $\xrv^n$ taking values in $\alphab{o}^m_l$ is upper bounded by
  \begin{equation}
  \frac{1}{n} \ent{\xrv^n}\leq \frac{1}{l} H_2\left(\frac{\avg{\frac{1}{n}\sum_{i=0}^{n-1}\xr_i}}{\alpha} \right),
  \end{equation}
  with equality when the $\xrv^l$ symbols forming $\xrv^n$ are i.i.d.
  We now recall that when $\xrv^n$ takes values in $\alphab{o}^m_l$ with $l > \beta/\alpha$ the input $\xv^n$ can be uniquely determined from the output sequence $\yv^n$ and $\ent{\xrv^n|\yrv^n} = 0$. We conclude the proof by selecting $l = \ceil{\wlen}$.
  \end{proof}

\section{Conclusion}
We have studied the information leakage rate of EMSs with finite battery capacity for general random processes modelling the energy consumption of the user. Inspired by the results on permuting channels we have proposed a battery charging policy with bounded information leakage rate for arbitrary random processes. We have particularized the analysis to the case in which the average energy consumption of the user is known and we have concluded that extreme values of the average energy consumption provide lower values of information leakage to the utility provider.






%

\balance
\bibliographystyle{IEEEbib}
\bibliography{thesisbiblio}

\end{document}